\documentclass[conference,a4paper]{IEEEtran}

\newcommand{\extra}[1]{}

\newcommand{\comment}[1]{}

\usepackage{amsmath,amssymb,amsfonts}
\usepackage{amsthm}
\usepackage{graphicx}

\newtheorem{theorem}{Theorem}[section]

\newtheorem{proposition}[theorem]{Proposition}

\theoremstyle{remark}

\def\squareforqed{\hbox{\rlap{$\sqcap$}$\sqcup$}}
\def\qed{\ifmmode\squareforqed\else{\unskip\nobreak\hfil
\penalty50\hskip1em\null\nobreak\hfil\squareforqed
\parfillskip=0pt\finalhyphendemerits=0\endgraf}\fi}
\def\endenv{\ifmmode\;\else{\unskip\nobreak\hfil
\penalty50\hskip1em\null\nobreak\hfil\;
\parfillskip=0pt\finalhyphendemerits=0\endgraf}\fi}

\renewenvironment{proof}{\noindent \textbf{{Proof~} }}{\qed\medskip}
\newenvironment{proof+}[1]{\noindent \textbf{{Proof #1~} }}{\qed\medskip}

\mathchardef\ordinarycolon\mathcode`\:
\mathcode`\:=\string"8000
\def\vcentcolon{\mathrel{\mathop\ordinarycolon}}
\begingroup \catcode`\:=\active
  \lowercase{\endgroup
  \let :\vcentcolon
  }

\newcommand{\nc}{\newcommand}
\nc{\rnc}{\renewcommand}
\nc{\beq}{\begin{equation}}
\nc{\eeq}{{\end{equation}}}
\nc{\beqa}{\begin{eqnarray}}
\nc{\eeqa}{\end{eqnarray}}
\nc{\lbar}[1]{\overline{#1}}
\nc{\bra}[1]{\langle#1|}
\nc{\ket}[1]{|#1\rangle}
\nc{\ketbra}[2]{|#1\rangle\!\langle#2|}
\nc{\braket}[2]{\langle#1|#2\rangle}
\nc{\proj}[1]{| #1\rangle\!\langle #1 |}
\nc{\avg}[1]{\langle#1\rangle}
\nc{\smfrac}[2]{\mbox{$\frac{#1}{#2}$}}
\nc{\tr}{\operatorname{tr}}
\nc{\tracedist}[1]{\Delta_{}\!\left( #1 \right)}
\nc{\fid}[1]{F\!\left( #1 \right)}

\nc{\ox}{\otimes}
\nc{\dg}{\dagger}
\nc{\dn}{\downarrow}
\nc{\cA}{{\cal A}}
\nc{\cB}{{\cal B}}
\nc{\cC}{{\cal C}}
\nc{\cD}{{\cal D}}
\nc{\cE}{{\mathcal E}}
\nc{\cF}{{\cal F}}
\nc{\cG}{{\cal G}}
\nc{\cH}{{\cal H}}
\nc{\cI}{{\cal I}}
\nc{\cJ}{{\cal J}}
\nc{\cK}{{\cal K}}
\nc{\cL}{{\cal L}}
\nc{\cM}{{\cal M}}
\nc{\cN}{{\cal N}}
\nc{\cO}{{\cal O}}
\nc{\cP}{{\cal P}}
\nc{\cR}{{\cal R}}
\nc{\cS}{{\cal S}}
\nc{\cT}{{\cal T}}
\nc{\cU}{{\cal U}}
\nc{\cV}{{\cal V}}
\nc{\cX}{{\cal X}}
\nc{\cZ}{{\cal Z}}

\nc{\entI}{{\bf I}}
\nc{\entIarrow}{{\bf I}^{\leftarrow}}
\nc{\entH}{{\bf H}}
\nc{\entS}{{\bf S}}
\nc{\entHmin}{\mathbf{H}_{\min}}
\nc{\entHtwo}{\mathbf{H}_{2}}
\nc{\aentHmin}{\hat{\mathbf{H}}_{\min}}
\nc{\binent}{h}

\nc{\supp}{\textrm{supp}}

\nc{\entF}{{\bf E}_f}

\nc{\isom}{\simeq}

\nc{\rank}{\operatorname{rank}}
\nc{\rar}{\rightarrow}
\nc{\lrar}{\longrightarrow}
\nc{\polylog}{\operatorname{polylog}}
\nc{\poly}{\operatorname{poly}}

\nc{\weight}{\textbf{w}}
\nc{\hamdist}{d_{H}}

\nc{\Sp}{{{\mathbb S}}}
\nc{\RR}{{{\mathbb R}}}
\nc{\CC}{{{\mathbb C}}}
\nc{\FF}{{{\mathbb F}}}
\nc{\NN}{{{\mathbb N}}}
\nc{\ZZ}{{{\mathbb Z}}}
\nc{\PP}{{{\mathbb P}}}
\nc{\QQ}{{{\mathbb Q}}}
\nc{\UU}{{{\mathbb U}}}
\nc{\OO}{{{\mathbb O}}}
\nc{\EE}{{{\mathbb E}}}
\nc{\id}{{\operatorname{id}}}
\nc{\1}{\id}

\nc{\qubitchannel}{\id_2}
\nc{\bitchannel}{\overline{\id}_2}

\nc{\be}{\begin{equation}}
\nc{\ee}{{\end{equation}}}
\nc{\bea}{\begin{eqnarray}}
\nc{\eea}{\end{eqnarray}}
\nc{\<}{\langle}
\rnc{\>}{\rangle}
\nc{\Hom}[2]{\mbox{Hom}(\CC^{#1},\CC^{#2})}
\nc{\rU}{\mbox{U}}

\nc{\ob}[1]{#1}

\newcommand{\eqdef}	{\stackrel{\textrm{def}}{=}}

\newcommand{\ex}[1]	{\mathbf{E}\left\{ #1 \right\}}
\newcommand{\exc}[2]	{\underset{#1}{\mathbf{E}}\left\{ #2 \right\}}
\newcommand{\pr}[1]	{\mathbf{P}\left\{ #1 \right\}}

\newcommand{\ceil}[1]	{\left\lceil #1 \right\rceil}

\nc{\unif}{\textrm{unif}}

\nc{\circuit}{\textrm{circ}}
\nc{\haar}{\textrm{haar}}
\nc{\clifford}{\textrm{clifford}}
\nc{\Mclifford}{\operatorname{M}_{\clifford}}
\nc{\Mcirc}{\operatorname{M}_{\circuit}}

\begin{document}

\sloppy

\title{Short random circuits define good quantum error correcting codes}

\author{
  \IEEEauthorblockN{Winton Brown}
  \IEEEauthorblockA{D\'{e}partement de Physique\\
    Universit\'{e} de Sherbrooke, Canada\\
    Email: winton.brown@usherbrooke.ca}
  \and
  \IEEEauthorblockN{Omar Fawzi}
  \IEEEauthorblockA{Institute for Theoretical Physics\\
    ETH Zuerich, Switzerland\\
    Email: ofawzi@phys.ethz.ch}
}



\maketitle


\begin{abstract}
We study the encoding complexity for quantum error correcting codes with large rate and distance. We prove that random Clifford circuits with $O(n \log^2 n)$ gates can be used to encode $k$ qubits in $n$ qubits with a distance $d$ provided $\frac{k}{n} < 1 - \frac{d}{n} \log_2 3 - \binent(\frac{d}{n})$. In addition, we prove that such circuits typically have a depth of $O( \log^3 n)$. 
\end{abstract}

\section{Introduction}
Error-correcting codes are fundamental objects with many theoretical and practical applications. In the context of quantum information, quantum error correcting codes allow the preservation of quantum data in the presence of noise; see \cite{Got97} for a reference. In addition to their application to reliable communication and computation, quantum error correcting codes have found applications in cryptography such as quantum key distribution \cite{SP00} and secret sharing \cite{CGL99}. 
In this paper, a quantum error correcting code is a subspace of the Hilbert space associated with an $n$-qubit space. Such a code has two important parameters: the number of qubits $k$ that can be encoded in this subspace and the distance $d$ which quantifies the number of errors that can be corrected by the code. It is desirable to have both $k$ and $d$ as large as possible. 

Our objective here is to understand how efficient the encoders of good quantum error correcting codes can be. This question has been studied for classical codes in several computation models, see e.g., \cite{BM05, GHKPV12}. In this paper, we work in the circuit model with a gate set composed of all two-qubit gates. In this model it is simple to see that to obtain a linear distance, a linear number of gates are needed and the depth has to be at least $\Omega(\log n)$. On the other hand, it is known that a large family of quantum error correcting codes known as stabilizer codes, which include many good codes, can have an encoder using only $O(n^2)$ gates \cite{CG97}. 

Specifically, the encoders we consider here are constructed by choosing a circuit at random with a given number of gates. Random quantum circuits have been well studied in the quantum information literature. Random quantum circuits of polynomial size are meant to be efficient implementations that inherit many useful properties of ``uniformly'' chosen unitary transformations, which are typically very inefficient. Most of the work has been in analyzing convergence properties of the random circuit model \cite{RM, ELL05, ODP07, HL09, BVPRL}. In some sense, we are here also interested in the convergence properties because our aim is to show that short random circuits define codes that are as good as the codes defined by completely random Clifford unitaries.

\subsection{Results}
We prove that random quantum circuits with $O(n \log^2 n)$ gates can be used to encode $k$ qubits in $n$ qubits with a distance $d$ provided $\frac{k}{n} < 1 - \frac{d}{n} \log_2 3 - \binent(\frac{d}{n})$. This is asymptotically the same as the distance of a code defined by a  random unitary from the complete Clifford group, which is known to achieve the quantum Gilbert-Varshamov bound \cite{CRSS97,Got97}. But a typical Clifford unitary is only known to be computable using a circuit with $\Omega(n^2)$ gates.

We also study another complexity measure for circuits which is the depth. The depth is simply the number of time steps needed to evaluate the circuit, keeping in mind that gates acting on disjoint qubits can be executed simultaneously. By parallelizing the random quantum circuit mentioned in the previous paragraph, we prove the existence of codes achieving the same parameters with an encoding of depth $O(\log^3 n)$ and size $O(n \log^2 n)$. We remark that it was proved in \cite{MN01} that the encoding and decoding operation of any stabilizer code can be implemented with quantum circuits of depth $O(\log n)$ using $O(n^2)$ qubits of ancilla, i.e., qubits in some fixed state that are restored to their original state at the end of the computation. In contrast, random quantum circuits do not use any ancilla qubits.

We say here a brief word on the proof. The first step of the proof is to relate the property of interest, which is the distance of the code, to the second moment operator of the random quantum circuit. This operator can then be studied as a Markov chain and the distance of the code translates to a property of the Markov chain. The convergence times of such Markov chains arising from the second order moments have been previously studied in \cite{ODP07, HL09}. However, these convergence times are not sufficient to prove the result we are aiming for and can only give useful bounds when $\Omega(n^2)$ gates are applied. Instead, we analyze the Markov chain in a finer way without using the spectral gap and bound the probabilities of going from a state $\ell$ to state $m$ within $O(n \log^2 n)$ steps as a function of $\ell$ and $m$.


\section{Preliminaries}
\subsection{Generalities}
The state of a pure quantum system is represented by a unit vector in a Hilbert space. Quantum systems  are denoted $A, B, C\dots$ and are identified with their corresponding Hilbert spaces. The Hilbert spaces we consider here will be  $n$-qubits spaces of the form $(\CC^2)^{\otimes n}$. 
A density operator is a Hermitian positive semidefinite operator with unit trace. The density operator associated with a pure state is abbreviated by omitting the ket and bra $\psi \eqdef \proj{\psi}$. 
For an introduction to quantum information, we refer the reader to \cite{NC00, Wil11}.


Throughout the paper, we use the Pauli basis to decompose operators. The $2 \times 2$ Pauli operators can be represented as follows:
\[
\sigma_0 = \left(
\begin{array}{cc}
1 & 0 \\
0 & 1 \\
\end{array}
\right) \qquad
\sigma_1 = \left(
\begin{array}{cc}
0 & 1 \\
1 & 0 \\
\end{array}
\right)
\]
\[
\sigma_2 = \left(
\begin{array}{cc}
0 & -i \\
i & 0 \\
\end{array}
\right)
\qquad
\sigma_3 = \left(
\begin{array}{cc}
1 & 0 \\
0 & -1 \\
\end{array}
\right).
\]
For a string $\nu \in \{0,1,2,3\}^n$, we define $\sigma_{\nu} = \sigma_{\nu_1} \otimes \cdots \otimes \sigma_{\nu_n}$. The support $\supp(\nu)$ of $\nu$ is simply the subset $\{i \in [n] : \nu_i \neq 0\}$ and the weight $w(\nu) = |\supp(\nu)|$. Any operator acting on $(\CC^2)^{\otimes n}$ can be represented as
\[
T = \frac{1}{2^n} \sum_{\nu \in \{0,1,2,3\}} \tr[\sigma_{\nu}T] \sigma_{\nu}.
\]

\subsection{Quantum error correcting codes}
As we are interested in the encoding complexity of quantum error correcting codes, we describe codes using the encoding operation. The encoding operation is a unitary transformation on an $n$-qubit space that we decompose as $A \otimes B$. It takes as input a $k$-qubit state $\ket{\psi} \in A$ (the state to be encoded) and the remaining input qubits are set to $\ket{0}^{\otimes n-k} \in B$. Such a code is called an $[n,k]$ quantum error correcting code. See Figure \ref{fig:encoding} for an illustration.

	\begin{figure}[h]
	\begin{center}
		\includegraphics[scale=0.7]{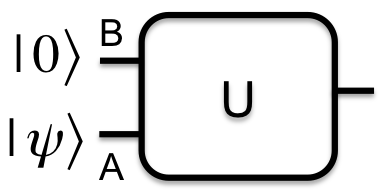}
		\caption{Encoding the state $\ket{\psi}$ using the encoder $U$}
		\label{fig:encoding}
	\end{center}
	\end{figure}

 For an encoding unitary $U$, the code space is defined as the vector space $\{ U \ket{\psi} \ket{0}^{\otimes n-k} : \ket{\psi} \in (\CC^2)^{\otimes k}\}$. By considering a basis $\{\ket{x} : x \in \{0,1\}^k\}$ of $A$, we obtain a basis $\{\ket{\bar{x}} = U \ket{x} \otimes \ket{0}^{\otimes n-k} : x \in \{0,1\}^{n-k}\}$ of the code space. A code has distance at least $d+1$ if we have for all $x, y \in \{0,1\}^k$ and all $\mu \in \{0,1,2,3\}^n$ with $1 \leq w(\mu) \leq d$
\begin{equation}
\label{eq:def-distance}
\bra{\bar{x}} \sigma_{\mu} \ket{\bar{y}} = C_{\mu} \delta_{xy}
\end{equation}
for some real numbers $C_{\mu}$ depending only on $\mu$ and not on $x,y$ \cite{BDSW96, KL97}. Here, $\delta_{xy} = 1$ if $x=y$ and zero otherwise. 
A code with minimum distance $2e+1$ can correct $e$ errors.

We can describe a unitary transformation (and more generally any superoperator) by describing how it acts on the Pauli basis, i.e., $U \sigma_{\nu} U^{\dagger}$ for all $\nu \in \{0,1,2,3\}^n$. This is especially convenient when talking about stabilizer codes for which the encoding operation $U$ belongs to the Clifford group. The Clifford group $\cC_n$ is defined as the set of unitary transformation $U$ under which the Pauli basis remains invariant up to phases. More precisely $\cC_n = \left\{ U : U \sigma_{\nu} U^{\dagger} \in \pm \{\sigma_{\mu}, \mu \in \{0,1,2,3\}^n\} \right\}$. Note that any unitary in $\cC_n$ can be implemented using Hadamard gates, phase gates and CNOT gates. When the encoding unitary $U \in \cC_n$, the minimum distance can be simply characterized as shown by the following proposition.
\begin{proposition}
\label{prop:distance}
A unitary $U \in \cC_n$ defines a quantum error correcting code of distance at least $d+1$ if and only if $d$ is the largest integer such that for all $\nu_A \in \{0,1,2,3\}^k - \{0\}$ and $\nu_B \in \{0,3\}^{n-k}$ and all $\mu \in \{0,1,2,3\}^n$ of weight $1 \leq w(\mu) \leq d$
\begin{equation}
\label{eq:no-small-paulis}
\tr[\sigma_{\mu} U \sigma_{\nu_A} \otimes \sigma_{\nu_B} U^{\dagger}] = 0.
\end{equation}
\end{proposition}
\begin{IEEEproof}
We start by proving the ``if'' part. We have
\begin{align}
\bra{\bar{x}} \sigma_{\mu} \ket{\bar{y}}
&= \bra{x}_A \otimes \bra{0}_B U^{\dagger} \sigma_{\mu} U \ket{y}_A \otimes \ket{0}_B \notag \\
&= \tr\left[ \sigma_{\mu} U \ketbra{y}{x} \otimes \proj{0} U^{\dagger} \right] \label{eq:ip-for-distance}
\end{align}
But
\[
\ketbra{y}{x} = \frac{\delta_{xy}}{2^k} \sigma_0 + \frac{1}{2^k}\sum_{\nu_A \in \{0,1,2,3\}^k, \nu_A \neq 0} \tr[\sigma_{\nu_A} \ket{y} \bra{x}] \sigma_{\nu_A},
\]
\begin{equation}
\label{eq:decompose-0}
\proj{0}^{\otimes n-k} = \frac{1}{2^{n-k}} \sum_{\nu_B \in \{0,3\}^{n-k}} \sigma_{\nu_B}.
\end{equation}
Plugging these expressions into \eqref{eq:ip-for-distance}, this shows that if condition \eqref{eq:no-small-paulis} holds, then condition \eqref{eq:def-distance} holds with $C_{\mu} = \tr\left[ \sigma_{\mu} U \frac{\sigma_{0}}{2^k} \otimes \proj{0} U^{\dagger} \right]$.

For the ``only if'' part, assume $U$ satisfies the condition \eqref{eq:def-distance}. Write
\[
\sigma_{\nu_A} = \sum_{x,y \in \{0,1\}^{k}} \sigma_{\nu_A}(x,y) \ketbra{x}{y}.
\]
Thus, if $1 \leq w(\mu) \leq d$,
\begin{align*}
&\tr[\sigma_{\mu} U \sigma_{\nu_A} \otimes \proj{0} U^{\dagger}] \\
&= \sum_{x,y \in \{0,1\}^{k}} \sigma_{\nu_A}(x,y) \tr\left[\sigma_{\mu} U \ketbra{x}{y} \otimes \proj{0} U^{\dagger}\right] \\
&= \sum_{x,y \in \{0,1\}^{k}} \sigma_{\nu_A}(x,y) C_{\mu} \delta_{xy} \\
&= 0,
\end{align*}
because $\tr[\sigma_{\nu_A}] = \sum_{x} \sigma_{\nu_A}(x,x) = 0$. Moreover, recall \eqref{eq:decompose-0} and note that  $U$ transforms Pauli operators to Pauli operators. As a result, we have that $\tr[\sigma_{\mu} U \sigma_{\nu_A} \otimes \proj{0} U^{\dagger}] = 0$ implies that for all $\nu_B \in \{0,3\}^{n-k}$, $\tr[\sigma_{\mu} U \sigma_{\nu_A} \otimes \sigma_{\nu_B} U^{\dagger}] = 0$.
\end{IEEEproof}
Consider a Pauli string $\nu_A \nu_B$ of weight for example $1$. A two-qubit gate can increase the weight of this Pauli string by at most $1$ and thus any code should have at least as many gates as its distance. Also as the weight of a Pauli string can be multiplied by at most two by a set of two-qubit gates acting on disjoint qubits, the depth of the encoding should be at least the logarithm of the distance.


\subsection{Random quantum circuits}
\label{sec:prelim-rqc}
We consider the following simple model for a quantum circuit acting on $n$ qubits. In a sequential random quantum circuit, a random two-qubit gate is applied to a randomly chosen pair of qubits in each time step. Here, the random two-qubit gate is going to be a random Clifford gate in $\cC_2$ acting on two qubits. 

A model of random circuits of a certain size defines a measure over unitary transformations on $n$ qubits that we call $p_{\circuit}$. 
The second-order moment operator will play an important role in all our proofs. The second-order moment operator is a superoperator acting on two copies of the space of operators acting on the ambient Hilbert space, which is an $n$-qubit space in our setting. For a measure $p$ over the unitary group, we can define the second moment operator $\operatorname{M}_p$ as
\[
\operatorname{M}_p[X \otimes Y] = \exc{U \sim p}{UXU^\dagger \otimes UYU^\dagger}.
\]
%
%
%
%
%
Even though we do not use the notion unitary designs here, it is worth pointing out that a distribution $p$ over unitary transformations is called a two-design if $\operatorname{M}_p = \Mclifford$ where $\clifford$ is the uniform distribution over the Clifford group \cite{DCEL09}.  We denote by $\Mcirc$ the moment operator for the distribution obtained by applying one step of the random circuit. It is possible to compute $\Mcirc$ explicitly when a random Clifford gate is applied to a randomly chosen pair $i,j$ of qubits, see e.g., \cite[Section 3.2]{HL09}. We have
\[
\Mcirc = \frac{1}{n(n-1)} \sum_{i \neq j} \operatorname{m}_{ij},
\]
where $\operatorname{m}_{ij}$ only acts on qubits $i$ and $j$ and is defined by
\[
\operatorname{m}_{ij}[\sigma_{\mu} \otimes \sigma_{\mu'}] = \left\{
\begin{array}{ll}
0 & \text{if } \mu \neq \mu' \\
\sigma_{0} \otimes \sigma_{0} & \text{if } \mu = \mu' = 0 \\
\frac{1}{15} \sum\limits_{\nu \in \{0,1,2,3\}^2, \nu \neq 0} \sigma_{\nu} \otimes \sigma_{\nu} &\text{if } \mu = \mu' \neq 0 \\
\end{array}
\right.
\]
for all $\mu, \mu' \in \{0,1,2,3\}^2$. We can thus represent the operator $\Mcirc$ in the Pauli basis using the following $4^n \times 4^n$ matrix
\[
Q(\mu, \nu) = \frac{1}{4^n} \tr\left[\sigma_{\nu} \otimes \sigma_{\nu} \Mcirc[\sigma_{\mu} \otimes \sigma_{\mu}] \right].
\]
In fact, it is simple to verify that $\sum_{\nu \in \{0,1,2,3\}^n} Q(\mu, \nu) = 1$ for all $\mu$ and so $Q$ can be seen as a transition matrix for a Markov chain over the Pauli strings of length $n$.

Now for a random circuit with $t$ independent random gates applied sequentially, the second moment operator is simply $\Mcirc^t$ and the corresponding matrix in the Pauli basis is also the $t$-th power of $Q$. The properties we are interested in can be expressed as quadratic functions of the entries of unitary transformation defined by the circuit and thus can be computed from the second moment operator. This means that these properties can be completely reduced to studying the evolution of the Markov chain defined by $Q$.

\section{Encoding using random quantum circuits}
\subsection{Sequential circuit}
The objective of this section is to prove the main result of the paper: a typical circuit with $t = O(n \log^2 n)$ gates defines an encoding into a quantum error correcting code with distance that is basically as good as for random Clifford unitaries. Let $\binent$ be the binary entropy function defined by $\binent(x) = -x \log_2 x - (1-x) \log_2(1-x)$.

\begin{theorem}[Good codes with almost linear-size circuits]
\label{thm:seq-good-codes}
For any $\delta > 0$, there exists a constant $c$ such that a random quantum circuit with $c n \log^2 n$ gates defines an $[n, k]$ quantum error correcting code with distance at least $d+1$ with probability at least
\[
1 - \frac{1}{n^8} - 2^{k- n \left(1 - \binent(d/n) - \log_2 (3) d/n - 3 \delta \right)}.
\]
\end{theorem}
\begin{proof}
We prove this result by using Proposition \ref{prop:distance}. Let $U_t$ denote the random unitary computed by choosing $t$ random gates. We can bound the probability that $U_t$ fails to satisfy condition \eqref{eq:no-small-paulis}.
\begin{align*}
&\pr{\exists \nu_A, \nu_B, \mu : w(\mu) \in [1,d],  \tr[\sigma_{\mu} U_t \sigma_{\nu_A} \otimes \sigma_{\nu_B} U_t^{\dagger}] \neq 0 } \\
&\leq \sum_{\nu_A, \nu_B} \sum_{\mu : w(\mu) \in [1,d]} \pr{\tr[\sigma_{\mu} U_t \sigma_{\nu_A} \otimes \sigma_{\nu_B} U_t^{\dagger}] \neq 0 }.
\end{align*}
Note that because $U_t \in \cC_n$, we have $\tr[\sigma_{\mu} U_t \sigma_{\nu_A} \otimes \sigma_{\nu_B} U_t^{\dagger}] = \pm 2^n$ when it is non-zero. This means that
\begin{align*}
&\pr{\tr[\sigma_{\mu} U_t \sigma_{\nu_A} \otimes \sigma_{\nu_B} U_t^{\dagger}] \neq 0 } \\
&= \frac{1}{2^n} \ex{\left|\tr[\sigma_{\mu} U_t \sigma_{\nu_A} \otimes \sigma_{\nu_B} U_t^{\dagger}] \right|} \\
&= \frac{1}{2^{2n}} \ex{\tr\left[\sigma_{\mu}^{\otimes 2} U_t^{\otimes 2} (\sigma_{\nu_A} \otimes \sigma_{\nu_B})^{\otimes 2} (U_t^{\dagger})^{\otimes 2} \right]} \\
&= \frac{1}{4^{n}} \tr\left[\sigma_{\mu}^{\otimes 2} \Mcirc^t\left[(\sigma_{\nu_A} \otimes \sigma_{\nu_B})^{\otimes 2}\right] \right] \\
&= Q^t(\nu, \mu),
\end{align*}
where $\nu = \nu_A \nu_B$ is the concatenation of $\nu_A$ and $\nu_B$. As a result, the probability that condition \eqref{eq:no-small-paulis} fails to hold can be bounded by
\begin{align}
&\sum_{\nu_A \in \{0,1,2,3\}^k, \nu_B \in \{0,3\}^{n-k}} \sum_{\mu: w(\mu) \in [1,d]} Q^t(\nu, \mu) \notag \\
&\quad = \sum_{\ell = 1}^{n} \sum_{\substack{\nu_A \in \{0,1,2,3\}^k, \nu_B \in \{0,3\}^{n-k} \\ w(\nu_A\nu_B) = \ell}} \sum_{m=1}^d \sum_{\mu: w(\mu) \in [1,d]} Q^t(\nu, \mu) \notag
 \\
&\quad = \sum_{\ell=1}^n \sum_{p=0}^{\ell} \binom{k}{p} 3^p \binom{n-k}{\ell - p} \sum_{m=1}^d P^t(\ell, m),
\label{eq:sum-probabilities}
\end{align}
where we defined the matrix $P(\ell, m) = \sum_{\mu : w(\mu) = m} Q(\nu, \mu)$ for any $\nu$ of weight $\ell$. It is not hard to see that this expression is independent of $\nu$. $P$ can be considered as the transition matrix of a Markov chain on $\{1, \dots, n\}$ and in fact, the probabilities $P(\ell, m)$ can be computed exactly (see \cite{HL09}):
\begin{equation}
\label{eq:def-p}
P(\ell,m) = \left\{
\begin{array}{ll}
1- \frac{2\ell(3n-2\ell-1)}{5n(n-1)} & \text{ if } m=\ell \\
\frac{2\ell(\ell-1)}{5n(n-1)} & \text{ if } m =\ell - 1\\
\frac{6\ell(n-\ell)}{5n(n-1)} & \text{ if } m = \ell + 1\\
0 & \text{ otherwise.}
\end{array} \right.
\end{equation}
In order to evaluate the expression in \eqref{eq:sum-probabilities}, we analyze the behaviour of the Markov chain when it runs for $O(n \log^2 n)$ steps. One possible route to bounding this expression would be to compute the mixing time of the chain defined by $P$. The stationary distribution for this chain is quite simple and it is defined as
\[
P_{\clifford}(m) = \frac{\binom{n}{m} 3^m}{4^n - 1}.
\]
The problem with this approach is that it only gives a useful bound whenever $t = \Omega(n^2)$. In fact, this is what is done in \cite{HL09}, which proves that random Clifford circuits with $O(n^2)$ gates are approximate two-designs. To prove our result for circuits of almost linear size, we need to analyze the chain more carefully. In short, the problem with the mixing time is that it is about the worst case starting point. For example, starting at the state $\ell =1$, it takes more time for the walk to mix, whereas if you start near the state $3n/4$, the distribution is already almost mixed. But an important point to realize is that the number of Pauli strings of low weight is small. So it is not necessary for the bound on $P^t(\ell, m)$ to be equally good for all $\ell$. And in fact by analyzing the Markov chain precisely, one can prove the following almost optimal bounds.
\begin{theorem}[\cite{BF13}]
\label{thm:mc-convergence}
Let $P$ be the transition matrix of the Markov chain defined in \eqref{eq:def-p}. For any constants $\delta \in (0,1/4), \eta \in (0,1)$, there exists a constant $c$ such that for $t \geq c n \log^2 n$ and all integers  $1 \leq \ell \leq n$ and $1 \leq m \leq 3n/4$, we have for large enough $n$
\begin{equation}
\label{eq:prob-ell-k}
P^t(\ell, m) \leq 4^{\delta n} \cdot \frac{\binom{n}{m} 3^m}{4^n - 1} + \frac{1}{(3-\eta)^{\ell} \binom{n}{\ell}}\frac{1}{n^{10}}.
\end{equation}
\end{theorem}
Let us state some remarks about this theorem. If instead of a random circuit, we were to apply a completely random Clifford unitary on $n$ qubits, we would have transition probabilities that are independent of $\ell$:
\[
P_{\clifford}(\ell, m) = \frac{\binom{n}{m} 3^m}{4^n - 1}.
\]
Also the dependence in $\ell$ is close to optimal. Starting at state $\ell$, there is a probability of roughly $\frac{1}{3^{\ell} \binom{n}{\ell}}$ to get back to the state $1$. And in state $1$, there is a probability of $2^{-O(\log^2 n)}$ of staying there for $t = O(n \log^2 n)$ steps.

In order to bound expression \eqref{eq:sum-probabilities}, we first evaluate $\sum_{p=0}^{\ell} \binom{k}{p} 3^p \binom{n-k}{\ell - p}$. We bound this sum simply by finding the $p$ for which it is maximum.
In order to do so, note that the ratio
\[
\frac{\binom{k}{p+1} 3^{p+1} \binom{n-k}{\ell - p-1}}{\binom{k}{p} 3^p \binom{n-k}{\ell - p}} = 3 \frac{(k-p) (\ell-p)}{(p+1)(n-k+p-\ell+1)}
\]
is a decreasing function of $p$. Moreover, if we plug $p = \lambda\ell$ with $\lambda = \frac{3k}{n+2k}$ in this expression, we obtain 
\[
\frac{6 k^2 \ell-9 k \ell^2+3 k \ell n}{4 k^2-2 k \ell+6 k^2 \ell-3 k \ell^2+4 k n-\ell n+3 k \ell n+n^2} \leq 1.
\]
%
%
This means that the maximum occurs for some $p_{\max} \leq \ceil{\lambda \ell}$. As a result, we can bound the sum by
\begin{align*}
\sum_{p=0}^{\ell} \binom{k}{p} 3^p \binom{n-k}{\ell - p} &\leq (\ell + 1) \cdot 3^{p_{\max}} \binom{k}{p_{\max}} \binom{n-k}{\ell-p_{\max}} \\
&\leq (\ell+1) \cdot 3^{\lambda \ell + 1} \sum_{p=0}^{\ell} \binom{k}{p} \binom{n-k}{\ell - p} \\
&= (\ell+1) \cdot 3^{\lambda \ell + 1} \binom{n}{\ell}.
\end{align*}
Note that if $k > (1-\delta) n$, there is nothing to prove because the probability bound given by the theorem is negative.
Assuming $k \leq (1-\delta)n$ and choosing $\eta$ appropriately small to apply Theorem \ref{thm:mc-convergence}, we obtain
\[
\sum_{p=0}^{\ell} \binom{k}{p} 3^p \binom{n-k}{\ell - p} \frac{1}{(3-\eta)^{\ell} \binom{n}{\ell}}\frac{1}{n^{10}}  \leq \frac{1}{n^{10}}.
\]
Also observe that
\[
\sum_{\ell=1}^n \sum_{p=0}^{\ell} \binom{k}{p} 3^p \binom{n-k}{\ell - p} \leq 2^{n+k} - 1.
\]
Combining this with \eqref{eq:prob-ell-k} and plugging this into \eqref{eq:sum-probabilities}, we obtain
\begin{align}
&\sum_{\ell=1}^n \sum_{p=0}^{\ell} \binom{k}{p} 3^p \binom{n-k}{\ell - p} \sum_{m=1}^d P^t(\ell, m) \notag \\
&\leq \sum_{\ell=1}^n \sum_{m=1}^d (2^{n+k} - 1) \cdot 4^{\delta n} \cdot \frac{\binom{n}{m} 3^m}{4^n - 1} + \frac{1}{n^{10}} \notag \\
&\leq \frac{1}{n^8} + n 4^{\delta n} 2^{k-n} \sum_{m=1}^d \binom{n}{m} 3^m \notag \\
&\leq \frac{1}{n^8} + 2^{k- n \left(1 - \binent(d/n) - \log_2 (3) d/n - 3 \delta \right)} \label{eq:prob-failure}
\end{align}
for sufficiently large $n$.
\end{proof}


\subsection{Parallelizing the circuit}
In the previous section, the complexity measure for a circuit was the number of two-qubit gates applied. Another important complexity measure for circuits is the depth. The depth is related to the time complexity, or the number of time steps needed in order to execute the circuit. Note that two gates acting on disjoint qubits could in fact be executed simultaneously. For the problem of finding good error correcting codes, a natural question is how small can the depth of encoding circuits be. In this section, we show that by parallelizing the random quantum circuits considered in the previous section, we obtain with high probability circuits with depth $O(\log^3 n)$. Using Theorem \ref{thm:seq-good-codes}, this proves that typical circuits of polylogarithmic depth define quantum error correcting codes that achieve a distance that is basically as good as the distance of a random stabilizer code.

To construct the parallelized circuit, one keeps adding gates to the current level until there is a gate that shares a qubit with a previously added gate in that level, in which case create a new level and continue. In the following proposition, we prove that by parallelizing a random circuit on $n$ qubits having $t$ gates we obtain with high probability a circuit of depth $O(\frac{t}{n} \log n)$.

\begin{proposition}[\cite{BF13}]
\label{prop:parallelization}
Consider a random sequential circuit composed of $t$ gates where $t$ is a polynomial in $n$. Then parallelize the circuit as described above. Except with probability $n^{-10}$, the resulting circuit has depth at most $O\left(\frac{t}{n} \log n\right)$.
\end{proposition}


\begin{theorem}[Good codes with low-depth circuits]
\label{thm:good-codes}
For any constant $\delta > 0$ there exist $[n,k]$ stabilizer codes with encoding circuits of depth $O(\log^3 n)$ and size $O(n \log^2 n)$ with a distance of $d$ provided $\frac{k}{n} \leq 1 - \binent(d/n) - \log_2 (3) d/n - 4 \delta $ and $n$ is large enough.
\end{theorem}
\begin{proof}
In the proof of Theorem \ref{thm:seq-good-codes}, we saw that with probability at most $1/3$, a random quantum circuit fails to define a circuit with distance as specified in the statement. Using Proposition \ref{prop:parallelization}, the probability that a random circuit with $t = O(n \log^2 n)$ gates leads to a large depth when parallelized is also at most $1/3$. We conclude that there exists a circuit for which both of these conditions hold.
\end{proof}

\section{Conclusion}
We have shown that good quantum error correcting codes can have encoding circuits with $O(n \log^2 n)$ gates and depth $O(\log^3 n)$. It is simple to show that $\Omega(n)$ gates are needed as well as a depth of $\Omega(\log n)$. It would be interesting to determine whether these simple lower bounds are achievable. It would also be interesting to study a random circuit model with a more restricted gate set, for example with the Hadamard, phase and CNOT gate. Can we obtain the same result for this gate set?

\section*{Acknowledgements}
We would like to thank Patrick Hayden and David Poulin for helpful discussions. We would also like to thank the anonymous referees for their suggestions and in particular for pointing out \cite{MN01}. The research of WB is supported by the Centre de Recherches Math\'ematiques at the University of Montreal, Mprime, and the Lockheed Martin Corporation. The research of OF is supported by the European Research Council grant No. 258932.

\bibliographystyle{IEEEtranS}
\bibliography{IEEEabrv,scrambling}

\begin{thebibliography}{10}
\providecommand{\url}[1]{#1}
\csname url@samestyle\endcsname
\providecommand{\newblock}{\relax}
\providecommand{\bibinfo}[2]{#2}
\providecommand{\BIBentrySTDinterwordspacing}{\spaceskip=0pt\relax}
\providecommand{\BIBentryALTinterwordstretchfactor}{4}
\providecommand{\BIBentryALTinterwordspacing}{\spaceskip=\fontdimen2\font plus
\BIBentryALTinterwordstretchfactor\fontdimen3\font minus
  \fontdimen4\font\relax}
\providecommand{\BIBforeignlanguage}[2]{{%
\expandafter\ifx\csname l@#1\endcsname\relax
\typeout{** WARNING: IEEEtranS.bst: No hyphenation pattern has been}%
\typeout{** loaded for the language `#1'. Using the pattern for}%
\typeout{** the default language instead.}%
\else
\language=\csname l@#1\endcsname
\fi
#2}}
\providecommand{\BIBdecl}{\relax}
\BIBdecl

\bibitem{BM05}
L.~Bazzi and S.~Mitter, ``Endcoding complexity versus minimum distance,''
  \emph{IEEE Trans. Inform. Theory}, vol.~51, no.~6, pp. 2103--2112, 2005.

\bibitem{BDSW96}
C.~H. Bennett, D.~DiVincenzo, J.~A. Smolin, and W.~K. Wootters, ``Mixed-state
  entanglement and quantum error correction,'' \emph{Phys. Rev. A}, vol.~54,
  no.~5, pp. 3824--3851, 1996, arXiv:quant-ph/9604024.

\bibitem{BF13}
W.~Brown and O.~Fawzi, ``Decoupling with random quantum circuits,'' 2013, in
  preparation. see arXiv:1210.6644 for slightly weaker results.

\bibitem{BVPRL}
W.~Brown and L.~Viola, ``Convergence rates for arbitrary statistical moments of
  random quantum circuits,'' \emph{Phys. Rev. Lett.}, vol. 104, p. 250501,
  2010.

\bibitem{CRSS97}
A.~R. Calderbank, E.~M. Rains, P.~W. Shor, and N.~J.~A. Sloane, ``Quantum error
  correction and orthogonal geometry,'' \emph{Phys. Rev. Lett.}, vol.~78, pp.
  405--408, 1997, arXiv:quant-ph/9605005.

\bibitem{CG97}
R.~Cleve and D.~Gottesman, ``Efficient computations of encodings for quantum
  error correction,'' \emph{Phys. Rev. A}, vol.~56, pp. 76--82, 1997,
  arXiv:quant-ph/9607030.

\bibitem{CGL99}
R.~Cleve, D.~Gottesman, and H.~Lo, ``How to share a quantum secret,''
  \emph{Phys. Rev. Lett.}, vol.~83, no.~3, pp. 648--651, 1999.

\bibitem{DCEL09}
C.~Dankert, R.~Cleve, J.~Emerson, and E.~Livine, ``{Exact and approximate
  unitary 2-designs and their application to fidelity estimation},''
  \emph{Phys. Rev. A}, vol.~80, no.~1, p. 12304, 2009, arXiv:quant-ph/0606161.

\bibitem{ELL05}
J.~Emerson, E.~Livine, and S.~Lloyd, ``Convergence conditions for random
  quantum circuits,'' \emph{Phys. Rev. A}, vol.~72, no.~6, p. 060302, 2005.

\bibitem{RM}
J.~Emerson, Y.~Weinstein, M.~Saraceno, S.~Lloyd, and D.~Cory, ``Pseudo-random
  unitary operators for quantum information processing,'' \emph{Science}, vol.
  302, no. 5653, pp. 2098--2100, 2003.

\bibitem{GHKPV12}
A.~G{\'a}l, K.~Hansen, M.~Kouck{\`y}, P.~Pudl{\'a}k, and E.~Viola, ``Tight
  bounds on computing error-correcting codes by bounded-depth circuits with
  arbitrary gates,'' in \emph{Proc. ACM STOC}, 2012, pp. 479--494.

\bibitem{Got97}
D.~Gottesman, ``Stabilizer codes and quantum error correction,'' Ph.D.
  dissertation, Caltech, 1997, arXiv:quant-ph/9705052.

\bibitem{HL09}
A.~Harrow and R.~Low, ``Random quantum circuits are approximate 2-designs,''
  \emph{Comm. Math. Phys.}, vol. 291, pp. 257--302, 2009, arXiv:0802.1919v3.

\bibitem{KL97}
E.~Knill and R.~Laflamme, ``Theory of quantum error-correcting codes,''
  \emph{Phys. Rev. A}, vol.~55, no.~2, p. 900, 1997.

\bibitem{MN01}
C.~Moore and M.~Nilsson, ``Parallel quantum computation and quantum codes,''
  \emph{SIAM J. Comput.}, vol.~31, no.~3, pp. 799--815, 2002.

\bibitem{NC00}
M.~Nielsen and I.~Chuang, \emph{Quantum computation and quantum
  information}.\hskip 1em plus 0.5em minus 0.4em\relax Cambridge University
  Press, 2000.

\bibitem{ODP07}
R.~Oliveira, O.~Dahlsten, and M.~Plenio, ``Generic entanglement can be
  generated efficiently,'' \emph{Phys. Rev. Lett.}, vol.~98, no.~13, p. 130502,
  2007.

\bibitem{SP00}
P.~W. Shor and J.~Preskill, ``{Simple Proof of Security of the BB84 Quantum Key
  Distribution Protocol},'' \emph{Phys. Rev. Lett.}, vol.~85, no.~2, pp.
  441--444, 2000, arXiv:quant-ph/0003004.

\bibitem{Wil11}
M.~Wilde, \emph{From Classical to Quantum Shannon Theory}, 2011,
  arXiv:1106.1445.

\end{thebibliography}


\end{document}